\documentclass[a4paper, 11pt]{article}

\title{Graded modal logic and counting message passing automata} 

\date{} 					

\usepackage{authblk}

\author[ \hspace{-1ex}]{Veeti Ahvonen}
\author[ \hspace{-1ex}]{Damian Heiman}
\author[ \hspace{-1ex}]{Antti Kuusisto}

\affil[ \hspace{-1ex}]{Tampere University, Finland}

\usepackage{graphicx} 

\usepackage{hyperref}
\usepackage{verbatim}

\usepackage[utf8]{inputenc}
\usepackage[T1]{fontenc}
\usepackage{amsmath, amssymb, amsthm}
\usepackage{tikz}
\usepackage{lipsum}
\usepackage{pdfpages}
\usepackage{multicol} 
\usepackage{stmaryrd}

\usepackage[textwidth=151mm, hcentering, vmargin={24mm,32.5mm}, foot=20mm]{geometry}

\usepackage{arydshln}
\usepackage{makecell}
\usepackage{mathtools}
\usepackage{colonequals}

\usepackage{marginnote}

\usepackage{bm}

\theoremstyle{plain}
\newtheorem{theorem}{Theorem}[section]
\newtheorem{lemma}[theorem]{Lemma}
\newtheorem{proposition}[theorem]{Proposition}
\newtheorem{corollary}[theorem]{Corollary}

\theoremstyle{definition}
\newtheorem{remark}[theorem]{Remark}

\newcommand{\N}{\mathbb N}
\newcommand{\Z}{\mathbb Z}

\newcommand{\bk}{\mathbf{k}}
\newcommand{\bl}{\mathbf{l}}

\newcommand{\bo}{\mathbf{o}}

\newcommand{\cA}{\mathcal{A}}

\newcommand{\cC}{\mathcal{C}}

\newcommand{\cK}{\mathcal{K}}

\newcommand{\cM}{\mathcal{M}}
\newcommand{\cN}{\mathcal{N}}

\newcommand{\cT}{\mathcal{T}}

\newcommand{\fm}{\mathfrak{m}}

\newcommand*{\abs}[1]{\lvert#1\rvert}   

\newcommand{\msc}{\mathrm{MSC}}

\newcommand{\md}{\mathrm{md}}

\newcommand{\cmmpa}{\mathrm{CMMPA}}

\newcommand{\dom}{\mathrm{DOM}}

\newcommand{\GML}{\mathrm{GML}}
\newcommand{\GMML}{\mathrm{GMML}}
\newcommand{\GFP}{\mathrm{GFP}}
\newcommand{\width}{\mathrm{width}}

\date{}

\begin{document}

\maketitle

\begin{center}
    { email: \url{firstname.lastname@tuni.fi}} 
    \vspace{1em}
\end{center}

\begin{abstract}
    We examine the relationship of graded (multi)modal logic to counting (multichannel) message passing automata with applications to the Weisfeiler-Leman algorithm. We introduce the notion of graded multimodal types, which are formulae of graded multimodal logic that encode the local information of a pointed Kripke-model. We also introduce message passing automata that carry out a generalization of the Weisfeiler-Leman algorithm for distinguishing non-isomorphic graph nodes. We show that the classes of pointed Kripke-models recognizable by these automata are definable by a countable (possibly infinite) disjunction of graded multimodal formulae and vice versa. In particular, this equivalence also holds between recursively enumerable disjunctions and recursively enumerable automata. We also show a way of carrying out the Weisfeiler-Leman algorithm with a formula of first order logic that has been augmented with Härtig's quantifier and greatest fixed points.
\end{abstract}

\section{Introduction}

In this paper we examine the relationship between the Weisfeiler-Leman algorithm and graded (multi)modal logic. The Weisfeiler-Leman algorithm \cite{weisfeiler1968reduction} is used in graph isomorphism testing. Informally, the algorithm starts with a $1$-coloring of a graph (depending on the source, the color of a node may also be its degree \cite{cai1992optimal} or just arbitrary \cite{barcelo}). The algorithm runs in discrete rounds; in each round, a node obtains a new color based on its previous color and how many neighbors of each color it has. The new color is determined injectively; if two nodes have different colors, or if they do not have the same number of neighbors of some color, then their new colors are also different. The algorithm stops when the coloring is no longer refined between rounds. 
Graded modal logic is a generalization of ordinary modal logic where instead of using typical modal diamonds $\Diamond \varphi$ we are allowed to use ``counting'' modalities $\Diamond_{\geq k} \varphi$, which basically state that ``$\varphi$ is true in at least $k$ neighbors''. This can be generalized to obtain graded multimodal diamonds $\langle \alpha \rangle_{\geq k}$.

The Weisfeiler-Leman algorithm bears some resemblance to type automata as presented in \cite{Kuusisto13}. 
Informally, a type automaton is a distributed system that refines a classification according to a node's own type (encoding of local information) and the set of types of its neighbors and ``accepts'' nodes of certain types. We present a modified type automaton that utilizes graded multimodal logic, where the classification is refined according to a node's type and the \emph{multiset} of types of its neighbors of each modality. This is essentially the Weisfeiler-Leman algorithm if it starts from an arbitrary coloring, takes different modalities into account and gives a binary classification of nodes on top of a coloring.

The type automaton is a specific type of message passing automaton which we call ``counting multichannel message passing automaton'' (or $\cmmpa$ for short) which is a generalization of the message passing automaton also presented in \cite{Kuusisto13}. We show that these automata have the same expressiveness as finite or countably infinite disjunctions of formulae of graded multimodal logic. In particular, we show that this equivalence holds between recursively enumerable automata and recursively enumerable disjunctions of formulae of graded multimodal logic. Informally, a recursively enumerable automaton is an automaton that can be ``simulated'' by a Turing machine, i.e., it can be ``finitized''. Also informally, a recursively enumerable disjunction of formulae of graded multimodal logic is a formula where all the disjuncts of the formula can be ``enumerated with a Turing machine'', i.e., the disjunction can be ``finitized''.
These are the main results of our paper. The connection to the Weisfeiler-Leman algorithm also provides a link between our work and graph neural networks.

\subsubsection*{Related work.} 

Research on the descriptive complexity of distributed systems began with \cite{weakpdc12},\cite{weak_models}. Hella et. al. provided logical characterizations of various weak classes of distributed algorithms. For instance, it was shown that graded modal logic ($\GML$) captures the class $\mathrm{MB}(1)$ of constant time distributed algorithms where nodes can distinguish between multisets of messages and no further, and where nodes broadcast the same message to each neighbor. This result in particular is closely related to the subject of this paper, as the Weisfeiler-Leman algorithm can distinguish precisely between multisets of colors and broadcasts in the same way. Another result of the paper was that graded multimodal logic ($\GMML$) captures the class $\mathrm{MV(1)}$ where nodes can send different messages to different neighbors instead of broadcasting the same message. This also relates to our paper as $\GMML$ is the logic under discussion in this paper.

The limitation to constant time algorithms was lifted in \cite{Kuusisto13}, which showed that finite message passing automata (FMPA) are characterized by modal substitution calculus ($\msc$). It was also shown in \cite{Kuusisto13} that the classes of pointed models recognizable by infinite message passing automata with no run-time limitations are the ones co-definable by a modal theory. We provide a generalization of this result by considering infinite message passing automata that can count the number of messages received and have multiple channels for passing messages, and graded multimodal logic with countably infinite disjunctions.
In related work, \cite{Kuusisto13} also showed that the aforementioned logic $\msc$ contains the $\mu$-fragment of the modal $\mu$-calculus in the finite setting, which was further shown to capture finite asynchronous message passing automata in \cite{reiter}. The expressiveness of distributed automata and the decidability of the emptiness problem for them was studied in \cite{Kuusisto20}.
Later in \cite{dist_circ_mfcs}, it was shown that $\msc$ also captures distributed computing with Boolean circuits and identifiers, and in \cite{ahvonen2023neural} it was shown that the diamond-free fragment of $\msc$ called substitution calculus captures neural networks.

Regarding the connection between the Weisfeiler-Leman algorithm and graph neural networks, \cite{morris2019weisfeiler} and \cite{xu2018powerful} showed that the ability of the class of aggregate-combine GNNs to distinguish graph nodes is at most as effective as the Weisfeiler-Leman algorithm. Further, \cite{barcelo} showed that these aggregate-combine GNNs are exactly as effective at distinguishing graph nodes as graded modal logic. The graph neural networks in this class have a finite number of layers, and thus have limited running time. Later, in \cite{grohe2023descriptive} this class of GNNs was characterized with a guarded fragment of first-order logic with counting and built-in relations.

\section{Preliminaries}

In this section, we cover the necessary preliminary concepts of graded multimodal logic, graded multimodal types, and counting multichannel message passing automata. First, we will clarify some of the language and notation we will be using. We let $\emptyset$ denote the empty set, and we let $\varepsilon$ denote the empty tuple. When we say that a set or disjunction is countable, we mean that it is either finite or countably infinite. We also let $[n]$ denote the set $\{1, \dots, n\}$.
We let $\mathrm{PROP} = \{p_i \mid i \in \N\}$ denote the countable set of proposition symbols. We will sometimes use metavariables $p,q, q_1, q_2$, and so on to refer to proposition symbols in the set $\mathrm{PROP}$. We will often use the symbol $\Pi \subseteq \mathrm{PROP}$ to refer to a finite subset of proposition symbols.

\subsection{Graded multimodal logic}
In this section we introduce graded multimodal logic ($\GMML$). We define the syntax and semantics of $\GMML$, as well as the concepts of modal depth and logical equivalence in the usual way. We present an intuitive definition for what it means for an infinite disjunction of $\GMML$-formulae to define a class of pointed Kripke-models. We also introduce the concept of the width of a graded multimodal formula as a vector of vectors.

Let $\Pi$ be a finite set of proposition symbols, and let $I = [a]$ be a set of indices ($a \in \Z_{+}$). A \textbf{$(\Pi, I)$-formula of graded multimodal logic} is a formula $\varphi$ defined over the language
\[
\varphi \coloncolonequals \top \mid p_{i} \mid \neg \varphi \mid (\varphi \land \varphi) \mid \langle\alpha\rangle_{\geq k} \varphi
\]
where $p_{i} \in \Pi$, $\alpha \in I$ and $k \in \N$. Logical connectives $\lor$, $\rightarrow$ and $\leftrightarrow$ are defined in the usual way.
The modal depth of a $(\Pi, I)$-formula of $\GMML$ is the number of nested diamonds in the formula and the width of a $(\Pi, I)$-formula of $\GMML$ is a vector of vectors of the maximum numbers $k$ that appear in the diamonds $\langle\alpha\rangle_{\geq k}$ in the formula for each $\alpha \in I$ at each modal depth. More formally, the \textbf{modal depth} $\md(\varphi)$ of a formula $\varphi$ is defined as follows:
\begin{itemize}
    \item If $\varphi = \top$ or $\varphi = p_{i} \in \Pi$, then $\md(\varphi) = 0$.
    \item If $\varphi = \neg \psi$ for some $(\Pi, I)$-formula $\psi$ of $\GMML$, then $\md(\varphi) = \md(\psi)$.
    \item If $\varphi = (\psi \land \theta)$ for some $(\Pi, I)$-formulae $\psi$ and $\theta$ of $\GMML$, then we have that $\md(\varphi) = \max(\md(\psi), \md(\theta))$.
    \item If $\varphi = \langle\alpha\rangle_{\geq k} \psi$ for some $\alpha \in I$, $k \in \N$ and some $(\Pi, I)$-formula $\psi$ of $\GMML$, then $\md(\varphi) = \md(\psi) + 1$.
\end{itemize}
Similarly, we define the \textbf{width} $\width(\varphi)$ of a formula $\varphi$ recursively as follows:
\begin{itemize}
    \item If $\varphi = \top$ or $\varphi = p_{i} \in \Pi$, then $\width(\varphi) = \varepsilon$ where $\varepsilon$ is the empty tuple.
    \item If $\varphi = \neg \psi$ for some $(\Pi, I)$-formula $\psi$ of $\GMML$, then $\width(\varphi) = \width (\psi)$.
    \item If $\varphi = (\psi \land \theta)$ for some $(\Pi, I)$-formulae $\psi$ and $\theta$ of $\GMML$ such that
    \begin{itemize}
        \item $\width(\psi) = (\bk_{1}, \dots, \bk_{m})$ where $\bk_{i} = (k_{i, 1}, \dots, k_{i, \abs{I}}) \in \N^{\abs{I}}$ and
        \item $\width(\theta) = (\bl_{1}, \dots, \bl_{n})$ where $\bl_{i} = (\ell_{i, 1}, \dots, \ell_{i, \abs{I}}) \in \N^{\abs{I}}$, then
    \end{itemize}
    \[
        \width(\varphi) = (\bo_{1}, \dots, \bo_{\max(m, n)}),
    \]
    where $\bo_{i} = (\max(k_{i, 1}, \ell_{i, 1}), \dots, \max(k_{i, \abs{I}}, \ell_{i, \abs{I}})) \in \N^{\abs{I}}$ and we agree that $k_{i, j} = 0$ for all $i > m$ and $\ell_{i, j} = 0$ for all $i > n$.
    \item If $\varphi = \langle\alpha\rangle_{\geq k} \psi$ for some $\alpha \in I$, $k \in \N$ and some $(\Pi, I)$-formula $\psi$ of $\GMML$ such that $\width(\psi) = (\bk_{1}, \dots, \bk_{m})$ (where $\bk_{1}, \dots, \bk_{m} \in \N^{\abs{I}}$), then
    \[
        \width(\varphi) = (\bk_{0}, \bk_{1}, \dots, \bk_{m})
    \]
    where $\bk_{0} = (\underbrace{0, \dots, 0}_{\alpha - 1}, k, \underbrace{0, \dots, 0}_{\abs{I} - \alpha})$.
\end{itemize}

A \textbf{Kripke-model of the vocabulary $(\Pi, I)$} (or $(\Pi, I)$-model) is a tuple of the form $M = (W, (R_{\alpha})_{\alpha \in I}, V)$, where $W \neq \emptyset$ is called the \textbf{domain}, $R_{\alpha} \subseteq W \times W$ is called an \textbf{accessibility relation for $\alpha$} and $V \colon \Pi \to \wp(W)$ is called the \textbf{valuation} of the model. A \textbf{pointed $(\Pi, I)$-model} is a pair $(M, w)$, where $M = (W, (R_{\alpha})_{\alpha \in I}, V)$ is a $(\Pi, I)$-model and $w \in W$. Furthermore, we let $\cN^{\alpha}(w)$ denote the set $\{\, v \in W \mid (w, v) \in R_{\alpha} \,\}$ of \textbf{$\alpha$-neighbors of $w$}. In this paper, we restrict ourselves to finite $(\Pi, I)$-models.

We define recursively what it means that a $(\Pi, I)$-formula $\varphi$ of $\GMML$ is \textbf{true} in a pointed $(\Pi, I)$-model $(M, w)$ (or that $(M, w)$ \textbf{satisfies} $\varphi$), denoted by $(M, w) \models \varphi$.
\begin{enumerate}
    \item $(M, w) \models \top$ always.
    \item $(M, w) \models p_{i}$ if and only if $w \in V(p_{i})$.
    \item $(M, w) \models \neg \psi$ if and only if $(M, w) \not\models \psi$.
    \item $(M, w) \models (\psi \land \theta)$ if and only if $(M, w) \models \psi$ and $(M, w) \models \theta$.
    \item $(M, w) \models \langle\alpha\rangle_{\geq k} \psi$ if and only if $\abs{\{\, v \in W \mid (w, v) \in R_{\alpha}, (M, v) \models \psi \,\}} \geq k$.
\end{enumerate}
Intuitively, the formula $\langle\alpha\rangle_{\geq k} \psi$ is true if and only if a node has at least $k$ $\alpha$-neighbors that satisfy $\psi$. For the sake of convenience, we define the abbreviation $\langle\alpha\rangle_{= k} \varphi$ that denotes the formula $\langle\alpha\rangle_{\geq k} \varphi \land \neg \langle\alpha\rangle_{\geq k + 1} \varphi$.
We say that two $(\Pi, I)$-formulae $\varphi$ and $\psi$ of $\GMML$ are \textbf{logically equivalent} ($\varphi \equiv \psi$) if $(M, w) \models \varphi \Leftrightarrow (M, w) \models \psi$ for all pointed $(\Pi, I)$-models $(M, w)$.

Assume that $\Pi$ is a finite set of proposition symbols, $I = [a]$ is a set of indices $(a \in \Z_{+})$ and $S$ is a countable set of $(\Pi, I)$-formulae of $\GMML$. We define the infinite disjunction of formulae $\varphi \in S$ as the formula $\bigvee_{\varphi \in S} \varphi$, and we define for all pointed $(\Pi, I)$-models $(M,w)$ that $(M,w) \models \bigvee_{\varphi \in S} \varphi$ if and only if $(M,w) \models \varphi$ for some $\varphi \in S$. Logical equivalence is extended to include countable disjunctions in the obvious way. Now let $\cC$ denote the class of all finite pointed $(\Pi, I)$-models. We say that a class $\cK \subseteq \cC$ is \textbf{definable by a countable disjunction of $(\Pi, I)$-formulae of $\GMML$} if there exists a countable set $S$ of $(\Pi, I)$-formulae $\varphi$ of $\GMML$ such that for all $(M, w) \in \cC$:
\[
    (M, w) \in \cK \text{ if and only if } (M, w) \models \bigvee_{\varphi \in S} \varphi.
\]
Finally, we say that a countable disjunction $\bigvee_{\varphi \in S} \varphi$ of $(\Pi, I)$-formulae of $\GMML$ is \textbf{recursively enumerable}
if the set $S$ is recursively enumerable.\footnote{Note that more precisely ``recursively enumerable disjunction'' is an abbreviation. More formally, it means that there exists a Turing machine which can enumerate all the disjuncts of such a formula.}

\subsubsection*{Closure properties}

Of course countable disjunctions of $\GMML$-formulae are closed under disjunction. To make our logic more practical it is easy to see that countable disjunctions of $\GMML$-formulae are also closed under conjunction (defined in the obvious way). 

\begin{proposition}
    Let $\Pi$ be a finite set of proposition symbols and $I = [a]$ a set of indices (where $a \in \Z_+$). Given two countable (possibly infinite) disjunctions of $(\Pi, I)$-formulae of $\GMML$ $\bigvee_{\varphi \in S} \varphi$ and $\bigvee_{\psi \in S'} \psi$, we have that $\bigvee_{\varphi \in S} \varphi \land \bigvee_{\psi \in S'} \psi$ is also logically equivalent to a countable disjunction of $(\Pi, I)$-formulae of $\GMML$.
\end{proposition}

\begin{proof}
    Applying De Morgan's laws it is easy to see that
    \[
    \bigvee_{\varphi \in S} \varphi \land \bigvee_{\psi \in S'} \psi \equiv \bigvee_{\psi \in S'} \Big( \bigvee_{\varphi \in S} (\varphi \land  \psi) \Big) \equiv \bigvee_{\psi \in S', \varphi \in S} (\varphi \land \psi).   
    \]
\end{proof}
    
It is also easy to see that countable disjunctions are not closed under negation.
Consider the case where $\Pi = \{p\}$ and $I = \{1\}$, and consider the countably infinite set $S = \{p, \langle 1 \rangle_{\geq 1} p, \langle 1 \rangle_{\geq 1} \langle 1 \rangle_{\geq 1} p, \dots\}$ of $(\Pi, I)$-formulae of $\GMML$ as well as its infinite disjunction $\bigvee_{\varphi \in S} \varphi$. This disjunction states that the proposition $p$ is reachable with the relation $R_{1}$. The negation of the infinite disjunction is $\neg \bigvee_{\varphi \in S} \varphi \equiv \bigwedge_{\varphi \in S} \neg \varphi$, which is a countable conjunction of $(\Pi, I)$-formulae. Furthermore, no countable disjunction of $(\Pi, I)$-formulae of $\GMML$ is equivalent to this conjunction (i.e., no countable disjunction can state that $p$ is not reachable).

\subsection{Graded multimodal types}

In this section we define the concept of a graded multimodal type and a \emph{full} graded multimodal type. The intuition is that a graded multimodal type is a formula of $\GMML$ that specifies a pointed model up to a certain width; no $\GMML$-formula bounded by that width can distinguish between two pointed models that satisfy the type in question. A full graded multimodal type is a graded multimodal type that specifies a pointed model up to a certain modal depth; no $\GMML$-formula bounded by that depth can distinguish between two models that satisfy the full type in question.

Let $\Pi$ be a finite set of proposition symbols, let $I = [a]$ be a set of indices ($a \in \Z_{+}$), and let $(M, w)$ be a pointed $(\Pi, I)$-model. The \textbf{graded multimodal $(\Pi, I)$-type of width $\bk$} ($\bk = (\bk_{1}, \dots, \bk_{n}) \in (\N^{\abs{I}})^{n}$ for some $n \in \N$) of $(M, w)$ (denoted by $\tau^{(M,w)}_{\bk}$) is defined recursively as follows. For $\bk = \varepsilon$, we define that
\[
    \tau^{(M,w)}_{\varepsilon} \colonequals \bigwedge\limits_{w \in V(p_{i})} p_{i} \land \bigwedge\limits_{w \notin V(p_{i})} \neg p_{i}.
\]
We assume canonical bracketing and ordering of conjuncts to ensure that each pointed model has exactly one type of each width. In the special case where $\Pi = \emptyset$, we define instead that $\tau^{(M,w)}_{\varepsilon} \colonequals \top$.
Assume we have defined the graded multimodal $(\Pi, I)$-type of width $\bk \colonequals (\bk_{1}, \dots, \bk_{n}) \in (\N^{\abs{I}})^{n}$ of all pointed models. Let $T_{\bk}$ denote the set of all graded multimodal $(\Pi, I)$-types of width $\bk$, and let $\bk_{0} = (k_{1}, \dots, k_{\abs{I}}) \in \N^{\abs{I}}$. We define the graded multimodal $(\Pi, I)$-type of width $(\bk_{0}, \bk_{1}, \dots, \bk_{n})$ of $(M,w)$ as follows. 
\[
    \begin{aligned}
        &\tau^{(M,w)}_{(\bk_{0}, \bk_{1}, \dots, \bk_{n})} \colonequals \tau^{(M,w)}_{\varepsilon} &&\land \bigwedge_{\ell = 1}^{k_{\alpha}-1} \{\, \langle\alpha\rangle_{= \ell} \tau \mid \alpha \in I, \tau \in T_{\mathbf{\bk}}, (M,w) \models \langle\alpha\rangle_{= \ell} \tau \,\} \\
        &&&\land \bigwedge \{\, \langle\alpha\rangle_{\geq k_{\alpha}} \tau \mid \alpha \in I, \tau \in T_{\bk}, (M,w) \models \langle\alpha\rangle_{\geq k_{\alpha}} \tau \,\} \\
        &&&\land \bigwedge \{\, \langle\alpha\rangle_{= \abs{\cN^{\alpha}(w)}} \top \mid \alpha \in I, k_{\alpha} > \abs{\cN^{\alpha}(w)} \,\}.
    \end{aligned}
\]
To clarify, the formula tells how many $\alpha$-neighbors (up to $k_{\alpha}$) of $w$ satisfy each type, and if the number of $\alpha$-neighbors is less than $k_{\alpha}$, then the formula tells the precise number of $\alpha$-neighbors as well. Notice that in this case the formula will not change if we increase the value of $k_{\alpha}$ past $\abs{\cN^{\alpha}(w)} + 1$. By doing this for each index $\alpha$ at every step of the recursion, we arrive at the full type defined below.

Similarly, the \textbf{full graded multimodal $(\Pi, I)$-type of modal depth $n$} of $(M,w)$ (denoted by $\tau^{(M,w)}_{n}$) is defined recursively over $n$. For $n = 0$, we have $\tau^{(M,w)}_{0} \colonequals \tau^{(M,w)}_{\varepsilon}$. Assume we have defined the full graded multimodal $(\Pi, I)$-type of modal depth $n$ of all pointed $(\Pi, I)$-models. The full graded multimodal $(\Pi, I)$-type of modal depth $n+1$ of $(M,w)$ is any type $\tau^{(M,w)}_{(\bk_{1}, \dots, \bk_{n+1})}$ where $\bk_{i} = (k_{i, 1}, \dots, k_{i, \abs{I}})$ and
\[
    k_{i, \alpha} > \max\{\, \abs{\cN^{\alpha}(v)} \mid (w,v) \in (\bigcup_{\alpha \in I} R_{\alpha})^{i-1} \,\}
\]
for all $i \in [n+1]$ and $\alpha \in I$.

Note that the class of full graded multimodal types is not a set of types in the ordinary sense. Given an arbitrary $(\Pi, I)$-formula of $\GMML$, we can not construct a logically equivalent finite disjunction of full graded multimodal $(\Pi, I)$-types even if we restrict them by modal depth. (This is because modal depth only restricts the length of the vector indicating the width of the formula; it does not restrict the size of the elements of the vector.) However, as we will see in Lemma \ref{Formula_to_types_GMML}, we can construct a logically equivalent recursively enumerable disjunction of full types instead.

\subsection{Message passing automata}

In this section we define the concepts of \emph{counting multichannel message passing automata} and \emph{counting multichannel type automata}. These are similar to the message passing automata and type automata defined in \cite{Kuusisto13}. The main difference is that counting multichannel automata receive a vector of \emph{multisets} of messages instead of a set of messages, the state of a node can depend on the type of edge along which each message was sent, and the counting multichannel type automata use graded multimodal types as states instead of ordinary modal types.

Given a set $X$, a \textbf{multiset} of $X$ is a function $M \colon X \to \N$. Intuitively, for each $x \in X$, $M(x)$ tells how many times $x$ appears in the multiset $M$. A multiset can be written by listing each instance of an element explicitly, and using double brackets to distinguish it from a set. For example, the sets $\{1, 1, 2\}$ and $\{1, 2\}$ are the same set, but the multisets $\{\{1, 1, 2\}\}$ and $\{\{1, 2\}\}$ are different, because $1$ appears in them a different number of times.

Given a multiset $M$, by $x \in^{k} M$ we mean that $x$ appears \emph{at least} $k$ times in multiset $M$. 
For example, if $M = \{\{1, 1, 2\}\}$, then $1 \in^{1} M$, $1 \in^{2} M$ and $2 \in^{1} M$, but $0 \notin^{1} M$, $1 \notin^{3} M$ and $2 \notin^{2} M$. Trivially, $x \in^{0} M$ for all $x$.

Given a set $X$, we let $\fm(X) = \{\, M \mid \text{$M$ is a multiset of $X$}\,\}$, i.e., $\fm(X)$ denotes the set of all multisets of $X$. For example, $\fm(\{1\}) = \{\emptyset, \{\{1\}\}, \{\{1, 1\}\}, \{\{1, 1, 1\}\}, \dots\}$.

Next we define message passing automata where the nodes of a network are capable of counting how many of each state they receive as a message along each type of edge, which we call counting multichannel message passing automata. We also specify a special message passing automaton that makes use of the full graded multimodal $(\Pi, I)$-types of modal depth $n$ from the previous section.
Let $\Pi$ be a finite set of proposition symbols and let $I = [a]$ be a set of indices ($a \in \Z_{+}$). We define the \textbf{counting multichannel message passing automaton for $(\Pi, I)$} (or $\cmmpa$) as a tuple $A = (Q, \pi, \delta, F)$, where
\begin{itemize}
    \item $Q \neq \emptyset$ is a countable (possibly infinite) set of \textbf{states},
    \item $\pi \colon \wp(\Pi) \to Q$ is an \textbf{initializing function},
    \item $\delta \colon (\fm(Q))^{\abs{I}} \times Q \to Q$ is a \textbf{transition function}, and
    \item $F \subseteq Q$ is a set of \textbf{accepting states}.
\end{itemize}
The function $\pi$ determines the initial state of a node depending on the set of proposition symbols true in that node. The function $\delta$ assigns the next state of a node based on the \emph{multiset} of states received as messages from neighbors along each edge type as well as the previous state of the node.
The set $F$ is the set of states whereby each pointed $(\Pi, I)$-model that visits a state in $F$ is accepted by the automaton.
We say that the automaton $A$ is \textbf{recursively enumerable} if the set $F$ is recursively enumerable and the function $\delta$ is computable. \footnote{Note that the automaton could be further finitized by replacing the set of accepting states with an accepting function $F \colon Q \to \{0, 1\}$.}

Formally, the \textbf{run} of a $\cmmpa$ $A = (Q, \pi, \delta, F)$ (for $(\Pi, I)$) is defined over a $(\Pi, I)$-model $M = (W, (R_{\alpha})_{\alpha \in I}, V)$ as a sequence of global configurations $f_{t} \colon W \to Q$. The \textbf{global configuration} of $A$ in $(M, w)$ at time $t \in \N$ is defined recursively as follows. For $t = 0$, we define that if $\Pi' \subseteq \Pi$ is the set of proposition symbols true in $(M, w)$, then $f_{0}(w) = \pi(\Pi')$. Now assume we have defined $f_{n}$. First, we define the multiset $N_{\alpha} = \{\{\, f_{n}(v) \mid (w, v) \in R_{\alpha} \,\}\}$ (i.e., $N_{\alpha}$ is the multiset of states received as messages by $w$ in round $n + 1$ across edges labeled $\alpha$; note that messages flow in the direction opposite to the relation $R_{\alpha}$). Then we define that $f_{n + 1}(w) = \delta(N_{1}, \dots, N_{\abs{I}}, f_{n}(w))$. We say that $A$ \textbf{accepts} a pointed $(\Pi, I)$-model $(M, w)$ (in round $t$) if it visits an accepting state in some round $t \in \N$, i.e., $f_{t}(w) \in F$.

Let $\Pi$ be a finite set of proposition symbols, let $I = [a]$ be a set of indices ($a \in \Z_{+}$), and let $\cC$ denote the class of all finite pointed $(\Pi, I)$-models. We say that a class $\cK \subseteq \cC$ is \textbf{recognizable by a counting multichannel message passing automaton} if there exists a $\cmmpa$ $A$ such that $A$ accepts a pointed model $(M,w) \in \cC$ if and only if $(M, w) \in \cK$.

We recall that in \cite{Kuusisto13} it was shown that a class of finite message passing automata (where each automaton receives a set of messages instead of a vector of multisets in every round and the set of states is finite) can be characterized by so called modal substitution calculus (or $\msc$). This characterization can be extended directly to our framework, i.e., \emph{finite} counting multichannel message passing automata are characterized by $\msc$ extended with graded multimodalities. Also other accepting conditions relating to automaton theory are possible to extend for these automata. For example, finite message passing automata with Büchi or Muller accepting conditions can be characterized by $\msc$ with the corresponding ``accepting conditions'', i.e., by defining suitable semantics for $\msc$ (and again these can be directly extended to our framework).

\subsection{The Weisfeiler-Leman algorithm and type automata}

Intuitively, a counting multichannel message passing automaton is a generalization of the \textbf{Weisfeiler-Leman algorithm} (or WL-algorithm). The WL-algorithm starts from a $1$-coloring of a graph. In each round, a node examines its current color and how many neighbors of each color it has, and changes to a corresponding color such that no two nodes obtain the same color if they have either a different color or a different number of neighbors of some color. Effectively, in each round $n \in \N$ the WL-algorithm partitions the class of all finite pointed models into isomorphism classes, which is to say that pointed models with the same color are isomorphic (up to depth $n$). The algorithm stops once the partitioning no longer changes between rounds, at which point two nodes sharing the same color are fully isomorphic. The coloring produced by the Weisfeiler-Leman algorithm is called the \textbf{stable coloring}; note that it is achieved in round $n$ at the latest, where $n$ is the size of the input graph/model. We will next define a \emph{counting multichannel type automaton} that intuitively carries out the WL-algorithm starting from an arbitrary coloring and across multiple edge types or modalities. The empty vocabulary $\Pi = \emptyset$ corresponds to starting from the classical $1$-coloring, and the index set $I = [1]$ (where relations $R_{1}$ are furthermore irreflexive and symmetric) corresponds to the classical simple graph where edges are not labelled.

A \textbf{counting multichannel type automaton} $A$ for $(\Pi, I)$ is a $\cmmpa$ defined as follows. Its set of states is exactly the set $\cT$ of all full graded multimodal $(\Pi, I)$-types. The initializing function $\pi \colon \wp(\Pi) \to \cT$ is defined such that if $\Pi' \subseteq \Pi$ is the set of proposition symbols true at $(M, w)$, then $\pi(\Pi') = \tau^{(M,w)}_{0}$. Let $N_{1}, \dots, N_{\abs{I}}$ be multisets of full graded multimodal $(\Pi, I)$-types of some modal depth $n$ and let $\tau$ be one such type. We define the state transition function $\delta \colon (\fm(\cT))^{\abs{I}} \times \cT \to \cT$,
\[
    \begin{aligned}
        &\delta(N_{1}, \dots, N_{\abs{I}}, \tau) = \tau_{0} &&\land \bigwedge_{k = 1}^{\abs{N_{\alpha}}} \{\, \langle\alpha\rangle_{= k} \sigma \mid \alpha \in I, \sigma \in^{k} N_{\alpha}, \sigma \notin^{k + 1} N_{\alpha} \,\} \\
        &&&\land \bigwedge \{\, \langle\alpha\rangle_{= \abs{N_{\alpha}}} \top \mid \alpha \in I \,\},
    \end{aligned}
\]
where $\tau_{0}$ is the type of modal depth $0$ that appears as a subformula of $\tau$, but not in the scope of any diamond.
Other inputs where $N_{1}, \dots, N_{\abs{I}}$ and $\{\tau\}$ contain types whose modal depths do not match are not used and we may define them as we please. The set of accepting states can be defined to be any subset of $\cT$.
It can be seen that the state of $A$ at $(M, w)$ in round $n$ is $\tau$ if and only if $\tau$ is the full graded multimodal $(\Pi, I)$-type $\tau^{(M,w)}_{n}$ of modal depth $n$ of $(M,w)$.

\begin{remark}\label{computability_remark}
    The transition function $\delta$ is computable in polynomial time: A Turing-machine scans the multisets $N_{1}, \dots, N_{\abs{I}}$ and for any $k \leq \abs{N_{\alpha}}$ constructs the formula $\langle\alpha\rangle_{= k} \sigma$ for every type $\sigma$ such that $\sigma \in^{k} N_{\alpha}$ and $\sigma \notin^{k+1} N_{\alpha}$. It also scans the previous type $\tau$ and constructs the formula $\tau_{0}$ by simply copying the propositional information from $\tau$. Finally, it observes the size of each multiset $N_{\alpha}$ and constructs the formula $\langle\alpha\rangle_{= \abs{N_{\alpha}}} \top$. This is clearly computable in polynomial time given that the whole input only needs to be scanned once and the operations performed on it are simple.
\end{remark}

\section{GMML-formulae capture CMMPA-recognizable classes}

In this section, we establish that the classes of pointed $(\Pi, I)$-models definable by a countable (possibly infinite) disjunction of $(\Pi, I)$-formulae of $\GMML$ are exactly the ones recognizable by counting multichannel message passing automata for $(\Pi, I)$. Furthermore, the classes defined by recursively enumerable disjunctions are exactly those recognized by recursively enumerable automata. We start by showing the claim from left to right, from formulae to automata. First we present a useful lemma, according to which all formulae of $\GMML$ have a logically equivalent recursively enumerable disjunction of full graded multimodal types of the same modal depth. Across this whole section we assume that $\Pi$ is a finite set of proposition symbols and $I = [a]$ is a set of indices where $a \in \Z_{+}$.

\begin{lemma}\label{Formula_to_types_GMML}
    Each $(\Pi, I)$-formula $\varphi$ of modal depth $n$ of $\GMML$ has a logically equivalent recursively enumerable disjunction of full graded multimodal $(\Pi, I)$-types of modal depth $n$.
\end{lemma}
\begin{proof}
    Let $\varphi$ be a $(\Pi, I)$-formula of $\GMML$ where $\md(\varphi) = n$, and let $T_{n}$ be the set of all full graded multimodal $(\Pi, I)$-types of modal depth $n$. Let $\Phi = \{\, \tau \in T_{n} \mid \tau \models \varphi \,\}$ and $\neg \Phi= \{\, \tau \in T_{n} \mid \tau \models \neg\varphi \,\}$, and let $\bigvee \Phi$ denote the disjunction of the types in $\Phi$.
    Obviously we have that $\Phi \cap \neg \Phi = \emptyset$ and $\bigvee \Phi \models \varphi$. To show that $\varphi \models \bigvee \Phi$, it suffices to show that $\Phi \cup \neg \Phi = T_{n}$.
    Assume instead that $\tau \in T_{n} \setminus (\Phi \cup \neg \Phi)$. Then there exist pointed $(\Pi, I)$-models $(M, w)$ and $(N, v)$ that satisfy $\tau$ such that $(M, w) \models \varphi$ and $(N, v) \models \neg \varphi$. Since $(M, w)$ and $(N, v)$ satisfy the same full graded multimodal $(\Pi, I)$-type of modal depth $n$, there can be no $(\Pi, I)$-formula of $\GMML$ of modal depth at most $n$ that distinguishes $(M, w)$ and $(N, v)$, but $\varphi$ is such a formula. Ergo, $\bigvee \Phi$ and $\varphi$ are logically equivalent.

    To see that $\bigvee \Phi$ is recursively enumerable, we can go through the types in the set $T_{n}$ in the order of increasing width. We start by considering types where no modality at any depth has width more than $1$, then $2$, and so forth. The number of such types is finite in every step. For each such type, we go through all pointed models $(M, w)$ where the type is true (where the shortest directed path from $w$ to any node $v$ has length at most $n$, since neither $\varphi$ nor $\bigvee \Phi$ can see past an $n$-hop neighborhood). Since the types are full types, they specify the maximum degree of each modality at each depth, and the number of such models is thus also finite. For each such model, we check whether the formula $\varphi$ is true; if it is true in all the models where a type is true, then that type is included in the disjunction. Thus, the disjunction $\bigvee \Phi$ is recursively enumerable.
\end{proof}

With this lemma, we are ready to prove the implication from left to right.

\begin{theorem}\label{thrm: gmml to cmmpa}
    Each class of finite pointed $(\Pi, I)$-models definable by a countable disjunction of $(\Pi, I)$-formulae of $\GMML$ is recognizable by a counting multichannel type automaton for $(\Pi, I)$. Moreover, if the disjunction is recursively enumerable, then so is the automaton.
\end{theorem}
\begin{proof}
    Assume we have a class $\cK$ of finite pointed $(\Pi, I)$-models definable by a countable disjunction $\psi \colonequals \bigvee_{\varphi \in S} \varphi$ of $(\Pi, I)$-formulae $\varphi$ of $\GMML$.
    By Lemma \ref{Formula_to_types_GMML}, each $\varphi \in S$ is logically equivalent with a recursively enumerable disjunction $\varphi^*$ of full graded multimodal $(\Pi, I)$-types of $\GMML$ such that $\md(\varphi^*) = \md(\varphi)$.
    We define a counting multichannel type automaton $A$ whose set of accepting states $F$ is the set of types that appear as disjuncts of $\varphi^*$ for any $\varphi \in S$. Now $(M, w) \models \psi$ if and only if $(M, w) \models \tau$ for some $\tau \in F$ if and only if the state of $(M,w)$ in round $\md(\tau)$ in $A$ is $\tau$. Ergo, $A$ accepts exactly the pointed $(\Pi, I)$-models in $\cK$.

    Now, assume the disjunction $\psi$ is recursively enumerable, i.e., the set $S$ is recursively enumerable. For each $\varphi \in S$ we have a recursively enumerable disjunction $\varphi^{*}$ of full graded multimodal types, and these types are the accepting states of the automaton. In other words, we have a recursively enumerable disjunction of recursively enumerable disjunctions. To show that the resulting set $F$ itself is recursively enumerable, we know that we can recursively enumerate the formulae $\varphi \in S$, and that for each $\varphi \in S$ we can recursively enumerate the types in $\varphi^{*}$. We can thus recursively enumerate all states in $F$ diagonally, by starting with the first type in the first disjunct, then the second type of the first disjunct and the first disjunct in the second disjunct, and so forth. The computability of $\delta$ was established already in Remark \ref{computability_remark}. Thus, the automaton $A$ is recursively enumerable.
\end{proof}

Next we prove the claim from right to left, from automata to formulae. We begin again with a lemma, which shows that all counting multichannel message passing automata agree on any two pointed models that share the same full graded multimodal type, up to the round indicated by the type's modal depth, and vice versa.

\begin{lemma}\label{Types_to_states_GMML}
    Two finite pointed $(\Pi, I)$-models $(M, w)$ and $(N, v)$ satisfy exactly the same full graded multimodal $(\Pi, I)$-type of modal depth $n$ if and only if they share the same state in each round (up to $n$) for each counting multichannel message passing automaton $A$ for $(\Pi, I)$.
\end{lemma}
\begin{proof}
    We show the claim by induction over $n$. Let $n = 0$. Two pointed $(\Pi, I)$-models $(M, w)$ and $(N, v)$ share the same full graded multimodal $(\Pi, I)$-type of modal depth $0$, if and only if they satisfy the exact same proposition symbols, if and only if each initializing function $\pi$ assigns them both the same initial state.
    
    Now assume the claim holds for $n$. Two pointed $(\Pi, I)$-models $(M, w)$ and $(N, v)$ satisfy the same full graded multimodal $(\Pi, I)$-type of modal depth $n + 1$ if and only if \textbf{1)} they satisfy the same full graded multimodal $(\Pi, I)$-type of modal depth $0$ and \textbf{2)} for each full graded multimodal $(\Pi, I)$-type $\tau$ of modal depth $n$, they have the same number of $\alpha$-neighbors that satisfy $\tau$ for each $\alpha \in I$. Now \textbf{1)} holds if and only if $(M, w)$ and $(N, v)$ satisfy the same proposition symbols, and by the induction hypothesis \textbf{2)} is equivalent to the $\alpha$-neighbors of $(M, w)$ and $(N, v)$ sharing (pair-wise) the same state in each round (up to $n$) for each $\alpha \in I$ for each $\cmmpa$. This is equivalent to $(M, w)$ and $(N, v)$ sharing the same full graded multimodal $(\Pi, I)$-type of modal depth $0$ not within the scope of a diamond in their state, and receiving the same multiset of states as messages in round $n$ for each $\alpha \in I$ in each $\cmmpa$. By the definition of the transition function, this is equivalent to $(M, w)$ and $(N, v)$ sharing the same state in round $n + 1$ for each $\cmmpa$.
\end{proof}

Now we are ready to show the direction from right to left.

\begin{theorem}\label{thrm: cmmpa to gmml}
    Each class of finite pointed $(\Pi, I)$-models recognizable by a counting multichannel message passing automaton for $(\Pi, I)$ is definable by a countable disjunction of $(\Pi, I)$-formulae of $\GMML$.
    Moreover, if the automaton is recursively enumerable, then so is the disjunction.
\end{theorem}
\begin{proof}
    Assume that the class $\cK$ of finite pointed $(\Pi, I)$-models is recognizable by the counting multichannel message passing automaton $A$. Let $\cT$ be the set of all full graded multimodal $(\Pi, I)$-types and let 
    $$
    \Phi = \{\, \tau^{(M,w)}_{n} \in \cT \mid \text{$A$ accepts the pointed $(\Pi, I)$-model $(M, w) \in \cK$ in round $n$} \,\}.
    $$
    We define the countable disjunction $\bigvee_{\tau \in \Phi} \tau$ and show that $(M, w) \models \bigvee_{\tau \in \Phi} \tau$ if and only if $A$ accepts $(M, w)$.
    
    If $(M, w) \models \bigvee_{\tau \in \Phi} \tau$, then $(M, w) \models \tau^{(N,v)}_{n}$ for some pointed $(\Pi, I)$-model $(N, v)$ accepted by $A$ in round $n$. This means that $(M,w)$ and $(N,v)$ satisfy the same full graded multimodal $(\Pi, I)$-type of modal depth $n$. By Lemma \ref{Types_to_states_GMML}, this means that $(M, w)$ and $(N, v)$ share the same state in $A$ in each round $\ell \leq n$. Since $A$ accepts $(N, v)$ in round $n$, $A$ also accepts $(M, w)$ in round $n$. Conversely, if $A$ accepts $(M, w)$ in round $n$, then $\tau^{(M,w)}_{n} \in \Phi$ and thus $(M, w) \models \bigvee_{\tau \in \Phi} \tau$.

    Now, to prove the second claim, assume that the automaton is recursively enumerable. The only thing we have to show is that $\Phi$ is recursively enumerable. Given a full type $\tau_n^{(M,w)} \in \cT$, we can easily construct the smallest (finite) tree model $T^{(M,w)}_n$ with the root $r$ such that $(T^{(M,w)}_n, r) \models \tau_n^{(M,w)}$, i.e., there exists a computable function $\cT \to S$, where $S =  \{\,(T^{(M,w)}_n, r) \mid \tau_n^{(M,w)} \in \cT\,\}$. Therefore, since $\cT$ is recursively enumerable, then the set $S =  \{\,(T^{(M,w)}_n, r) \mid \tau_n^{(M,w)} \in \cT\,\}$ is also. To check if $\tau_n^{(M,w)} \in \Phi$, we simply first construct $(T^{(M,w)}_n,r)$ and feed it to the automaton. Since the automaton is recursively enumerable then $\Phi$ is also.
\end{proof}

Having proved the implication in both directions, we have shown the desired equivalence between classes of finite pointed models definable by countable disjunctions of $\GMML$-formulae and those recognizable by a $\cmmpa$.\footnote{We note that a similar characterization could be obtained for infinite models by using any graded multimodal types instead of only full types and modifying the automata accordingly, but we leave this for possible future work.} Combining both directions, we obtain the following corollaries.

\begin{corollary}\label{crl: type cmmpa and gmml}
    Each class of finite pointed $(\Pi, I)$-models definable by a countable disjunction of $(\Pi, I)$-formulae of $\GMML$ is recognizable by a counting multichannel message passing automaton for $(\Pi, I)$, and vice versa. Moreover, the countable disjunction is recursively enumerable if and only if the automaton is as well.
\end{corollary}

\begin{corollary}\label{crl: cmmpa to type cmmpa}
    Each class of finite pointed $(\Pi, I)$-models recognizable by a counting multichannel message passing automaton for $(\Pi, I)$ can also be recognized by a counting multichannel type automaton for $(\Pi, I)$. Moreover, if the counting multichannel message passing automaton is recursively enumerable, then so is the counting multichannel type automaton.
\end{corollary}

Note that by Remark \ref{computability_remark} and Corollaries \ref{crl: type cmmpa and gmml} and \ref{crl: cmmpa to type cmmpa}, we can translate both recursively enumerable disjunctions and recursively enumerable automata into ``syntactically decidable type automata'', i.e., 
recursively enumerable type automata with a transition function computable in polynomial time.

In the case where $\Pi = \emptyset$ and $I = [1]$, we also receive the following corollary.

\begin{corollary}
    Each class of finite pointed $(\emptyset, [1])$-models definable by a countable disjunction of $(\emptyset, [1])$-formulae of $\GMML$ is recognizable by the WL-algorithm, and vice versa.
\end{corollary}
Note that when we say that a class of pointed models is recognizable by the WL-algorithm, we mean the following. As the WL-algorithm partitions pointed models into isomorphism classes, two pointed models in the same isomorphism class are either both in the recognized class, or neither of them is.


\section{Expressing WL via first-order logic with Härtig's quantifier and greatest fixed points}

For an alternative, finite formula, that carries out the WL-algorithm, consider first-order logic. We amend the logic with Härtig's quantifier and greatest fixed points.

We denote the first-order structures with $\cM$, $\cN$, $\ldots$ and so on. The domain of the structure $\cM$ is denoted by $\dom(\cM)$. 
We let $H$ denote Härtig's quantifier. If $\varphi$ and $\psi$ are $\mathrm{FO}$-formulae and $x$ and $y$ are variable symbols, then $Hxy[\varphi, \psi]$ is a formula. We define that a model $\cM$ satisfies this formula if the number of elements in the domain of the model that satisfy $\varphi$ is the same as the number of elements that satisfy $\psi$. More formally, $\abs{\{\, a \in \dom(\cM) \mid \cM \models \varphi(a) \,\}} = \abs{\{\, b \in \dom(\cM) \mid \cM \models \psi(b) \,\}}$. 

For greatest fixed points, let $\varphi(x,y)$ be a $\mathrm{FO}$-formula (which can also contain Härtig's quantifiers), and let $W$ be a $2$-ary relation symbol. In this case, $\GFP_{W, x, y}[\varphi]$ is a formula. We define the truth of such a formula in a model $\cM$ as follows. Let $W^{0} \colonequals \dom(\cM)^{2}$. Now, assume we have defined the relation $W^{n}$ for some $n \in N$. We let $W^{n+1}$ denote the set of pairs $(a,b) \in \dom(M)^{2}$ such that $M \models \varphi(a,b)$, where each appearance of $W$ in $\varphi$ is interpreted as $W^{n}$. Now $\cM \models \GFP_{W, x, y}[\varphi](a,b)$ if and only if $(a,b) \in W^{n}$ for all $n \in \N$.

First, we informally describe (and later prove formally) how the Weisfeiler-Leman algorithm is described by a single formula. The formula is defined as follows:
\[
   \varphi_{\mathsf{WL}}(x,y) \colonequals \mathrm{GFP}_{W, x, y} [Wxy \land \forall z ( Hyx[ Exy \land Wyz, Eyx \land Wxz])], 
\]
where $E$ is the $2$-ary edge relation over the model according to which the WL-algorithm is run. The subformula $Hyx[ Exy \land Wyz, Eyx \land Wxz]$ checks that $x$ and $y$ have the same number of neighbors that have the same particular color. The universal quantification checks that this holds for all colors, and the formula $Wxy$ checks that $x$ and $y$ also had the same color in the previous round. In the context of this formula, the relation $W^{n}$ from the previous paragraph is the equivalence relation induced by the WL-algorithm over the elements of $\cM$ after $n$ rounds. By checking the truth of this formula iteratively indefinitely via the greatest fixed point, we clearly arrive at the same outcome as the WL-algorithm.

Next we formally show that the formula $\varphi_{\mathsf{WL}}(x,y)$ defines the Weisfeiler-Leman algorithm. Before that we define a few things.
Given two models $\cM$ and $\cN$, we let $\cM \uplus \cN$ denote the \textbf{disjoint union of the models} $\cM$ and $\cN$.
Given a model $\cM$ and $w \in \dom(\cM)$, we let $\mathsf{WL}^{\infty}(\cM,w)$ denote the color of the point $w$ in $\cM$ in the stable coloring of $\cM$ produced by the WL-algorithm. Moreover, if Weisfeiler-Leman stops with an input model $\cM$ in round $m$, then for all $n \leq m$ we let $\mathsf{WL}^n(\cM,w)$ denote the color of the point $w$ in $\cM$ in round $n$ produced by the Weisfeiler-Leman algorithm, and if $n > m$ we let $\mathsf{WL}^n(\cM, w) = \mathsf{WL}^{\infty}(\cM,w)$. 

\begin{theorem}
    For any two finite models $\cM$ and $\cN$ and points $w$ in $\cM$ and $u$ in $\cN$, we have
    \[
    \mathsf{WL}^{\infty}(\cM, w) = \mathsf{WL}^{\infty}(\cN,u) \iff \cM \uplus \cN \models \varphi_{\mathsf{WL}}(w, u).
    \]
\end{theorem}

\begin{proof}
    We let $\cA = \cM \uplus \cN$. 
    Recall that
    \[
    \varphi_{\mathsf{WL}}(x, y) \colonequals \mathrm{GFP}_{W, x y} [Wxy \land \forall z ( Hyx[ Exy \land Wyz, Eyx \land Wxz])].
    \]
    We prove by induction on $t \in \N$ that for each pointed model $(\cM, w)$ and $(\cN, u)$ we have
    \[
    \mathsf{WL}^{t}(\cM, w) = \mathsf{WL}^{t}(\cN,u) \iff  (w, u) \in W^t.
    \]
    The case for $t = 0$ is straightforward. Assume that the claim holds for $t = n$. We prove the claim for $t = n+1$. 
    First of all, it is easy to see that if $\mathsf{WL}^{n+1}(\cM, w) = \mathsf{WL}^{n+1}(\cN, u)$, then $\mathsf{WL}^{n}(\cM, w) = \mathsf{WL}^{n}(\cN, u)$. Also by the definition of $\varphi_{\mathsf{WL}}(x,y)$, we have $W^{n+1} \subseteq W^n $. 
    By the induction hypothesis, we have $\mathsf{WL}^{n}(\cM, w) = \mathsf{WL}^{n}(\cN, u)$ iff $(w, u) \in W^n$. 
    Moreover, by the induction hypothesis and by the definitions of $\varphi_{\mathsf{WL}}(x,y)$ and the $\mathsf{WL}$-algorithm, for each point $a \in \dom(\cA)$ with $\mathsf{WL}^n(\cA,a) = c$ we have
    \[
    \begin{aligned}
    \abs{\{\, w' \mid (w, w') \in E, \mathsf{WL}^n(\cM, w') = c \,\}}& =  \abs{\{\, u' \mid (u, u') \in E, \mathsf{WL}^n(\cN, u') = c \,\}} \\
    &\iff \\
    \abs{\{\, w' \mid (w, w') \in E, (w',a) \in W^n \,\}}& =  \abs{\{\, u' \mid (u, u') \in E, (u',a) \in W^n \,\}}.
    \end{aligned}
    \]
    In other words, $(w,u) \in W^{n+1}$ iff in round $n$ of the $\mathsf{WL}$-algorithm, $w$ and $u$ have the same color and for each color $c$ there are the same amount of neighbors of $w$ with color $c$ and neighbors of $u$ with color $c$.
    Therefore, we conclude
    \[
    (w, u) \in W^{n+1} \iff \mathsf{WL}^{n+1}(\cM, w) = \mathsf{WL}^{n+1}(\cN,u).
    \]
\end{proof}

It is worth noting that the disjoint union is important for the formula $\varphi_{\mathsf{WL}}(w, u)$, as it does not give names or specifications for the equivalence classes that it produces. If run separately in two different models, we would have no way of comparing the equivalence classes produced for the two models. This also helps to illustrate how capturing classification conditions is more difficult than merely splitting the domain into equivalence classes.

\section{Conclusion}

We have shown an equivalence between classes of pointed models definable with a countable disjunction of GMML-formulae and those recognizable by a specific type of message passing automaton. Furthermore, we have shown an equivalence between classes definable by recursively enumerable disjunctions and those recognizable by recursively enumerable automata. As a corollary we have shown an equivalence between classes definable by a specific kind of countable disjunction and those recognizable by the Weisfeiler-Leman algorithm. We have also established an alternative logic in which the Weisfeiler-Leman algorithm can be expressed using first-order logic, Härtig's quantifier and greatest fixed points. Future work could include refining this logic, or expanding the former characterization for infinite Kripke-models.

\subsection*{Acknowledgments} 
Antti Kuusisto was supported by the Academy of Finland project Theory of computational logics, grant numbers 352419, 352420, 353027, 324435, 328987. Damian Heiman was supported by the same project, grant number 353027. Antti Kuusisto was also supported by the Academy of Finland consortium project Explaining AI via Logic (XAILOG), grant number 345612. Veeti Ahvonen was supported by the Vilho, Yrjö and Kalle Väisälä Foundation of the Finnish Academy of Science and Letters.

\bibliography{references}

\end{document}